\documentclass[a4paper,11pt]{memoir}

\usepackage[english]{babel}
\usepackage[OT1]{fontenc}
\usepackage[utf8]{inputenc} %Heidi: changed utf8x to utf8 since my latex version was not happy with it
\usepackage{microtype}

\usepackage{amsmath,amsthm,amsfonts}
\usepackage{mathpazo}

\chapterstyle{tandh}
\makepagestyle{chapter}
\makeoddfoot{chapter}{}{}{\thepage}

\makepagestyle{ruled2}
\makepsmarks{ruled2}{%
  \nouppercaseheads
  \createmark{chapter}{left}{shownumber}{}{. \space}
  \createmark{section}{right}{nonumber}{}{}
  \createplainmark{toc}{both}{\contentsname}
  \createplainmark{lof}{both}{\listfigurename}
  \createplainmark{lot}{both}{\listtablename}
  \createplainmark{bib}{both}{\bibname}
  \createplainmark{index}{both}{\indexname}
  \createplainmark{glossary}{both}{\glossaryname}
}
\makeoddhead{ruled2}{\lsstyle\scshape\leftmark}{}{\itshape\rightmark}
\makeevenhead{ruled2}{\lsstyle\scshape\leftmark}{}{\itshape\rightmark}
\makeheadrule{ruled2}{\textwidth}{\normalrulethickness}
\makeoddfoot{ruled2}{}{}{\thepage}
\makeevenfoot{ruled2}{\thepage}{}{}

\makeatletter

\pretitle{\vspace{0pt plus 0.7fill}\begin{center}\HUGE}
\posttitle{\end{center}\par}
\preauthor{\par\begin{center}\let\and\\\Large}
\postauthor{\end{center}}
\predate{\par\begin{center}\Large}
\postdate{\end{center}}

\def\@advisors{}
\newcommand{\advisors}[1]{\def\@advisors{#1}}
\def\@department{}
\newcommand{\department}[1]{\def\@department{#1}}
\def\@thesistype{}
\newcommand{\thesistype}[1]{\def\@thesistype{#1}}

\renewcommand{\maketitlehookb}{\vspace{1in}%
  \par\begin{center}\Large\@thesistype\end{center}}

\renewcommand{\maketitlehookd}{%
  \vfill\par
  \begin{flushright}
    \@advisors\par
    \@department, ETH Z\"urich
  \end{flushright}
}

\makeatother

\checkandfixthelayout

\usepackage{pst-pdf}
\usepackage{cite}
\usepackage{paralist}
\usepackage{mathtools}

\newcommand{\naturals}{\mathbold{N}}
\newcommand{\integers}{\mathbold{Z}}
\newcommand{\nin}{\not\in}

\newtheoremstyle{thm}{}{}{}{}{\bfseries}{.}{0.5em}{\thmname{#1}\thmnumber{ #2}\thmnote{ #3}}
\newtheoremstyle{thmn}{}{}{\bfseries}{}{}{.}{0.5em}{\textit{\thmname{#1}\thmnumber{ #2}} (\thmnote{#3})}
\theoremstyle{thm}

\DeclarePairedDelimiter{\ceil}{\lceil}{\rceil}
\DeclarePairedDelimiter{\floor}{\lfloor}{\rfloor}
\DeclarePairedDelimiter{\abs}{\lvert}{\rvert}

\begin{document}

\newcommand{\newn}[1]{\Gamma_{\mathrm{new}}(#1)}

\newtheorem{thm}{Theorem}
\newtheorem{corollary}{Corollary}
\newtheorem{lemma}{Lemma}
\newtheorem{proposition}{Proposition}
\theoremstyle{definition}
\newtheorem{define}{Definition}

\title{Rainbow Cycles and Paths}
\author{Frank Mousset}
\date{August 22, 2011}
\thesistype{Bachelor Thesis}
\advisors{Advisors: Prof.\ Dr.\ Emo Welzl, Heidi Gebauer}
\department{Department of Computer Science}

\calccentering{\unitlength}

\frontmatter
\begin{adjustwidth*}{\unitlength}{-\unitlength}
%\begin{center}
%Rainbow Cycles and Paths
%\end{center}
%\thispagestyle{empty}
%\cleartorecto
\maketitle
\thispagestyle{empty}
\cleartorecto

\begin{abstract}
  In a properly edge colored graph, a subgraph
  using every color at most once is called \emph{rainbow}.
  In this thesis, we study rainbow cycles and paths in proper edge colorings of complete graphs,
  and we
  prove that in every proper edge coloring of $K_n$, there is a
  rainbow path on $(3/4-o(1))n$ vertices, improving on the previously
  best bound of $(2n+1)/3$ from~\cite{GyarfasMhalla2010}.

  Similarly, a $k$-rainbow path in a proper edge coloring of $K_n$
  is a path using no color more than $k$ times.
  We prove that in every proper edge coloring of $K_n$, there is a $k$-rainbow path on
  $(1-2/(k+1)!)n$ vertices.
\end{abstract}

\thispagestyle{chapter}
\end{adjustwidth*}
\cleartorecto
\tableofcontents

\mainmatter

\chapter{Introduction}

\section{Rainbow cycles and paths}

Consider an edge colored graph $G$.
A subgraph of $G$ is called \emph{rainbow} (or \emph{heterochromatic})
if no two of its edges receive the same color. We are concerned with
rainbow paths and, to a lesser extent, cycles in proper edge colorings of the complete
graph $K_n$.
Hahn conjectured that every proper edge coloring of $K_n$ admits a
Hamiltonian rainbow path (a rainbow path visiting every vertex of $K_n$) (c.f.~\cite{Maamoun1984}).
Maamoun and Meyniel~\cite{Maamoun1984} disproved this conjecture
by constructing counterexamples for the case where $n$ is a power of two, as follows.
Let $n=2^m$. Then we can identify the vertices of $K_n$ with distinct elements of the group
$(\integers/2\integers)^m$, and color every edge $\{a,b\}$ of $K_n$
with the sum of the group elements corresponding to $a$ and $b$.
This is a proper edge coloring, because for two edges $\{a,b\}$ and $\{a,c\}$,
the group property implies $a+b\neq a+c$. Maamoun and Meyniel proved that
this coloring admits no Hamiltonian rainbow paths.
The reader is invited
to check this fact for the case of $K_4$.

\begin{center}
  \begin{pspicture}(-4,-2.5)(4,3.5)
    \psset{fillstyle=solid,fillcolor=black}
    \cnode(0,0){.08}{a} \nput{30}{a}{$00$}
    \cnode(0,3){.08}{b} \nput{0}{b}{$01$}
    \cnode(-2.5891,-1.5,){.08}{c} \nput{180}{c}{$11$}
    \cnode(2.5981,-1.5){.08}{d} \nput{0}{d}{$10$}
    \ncline{a}{b}\naput{$01$}
    \ncline{a}{c}\naput{$11$}
    \ncline{a}{d}\nbput{$10$}
    \ncline{b}{c}\nbput{$10$}
    \ncline{b}{d}\naput{$11$}
    \ncline{c}{d}\nbput{$01$}
  \end{pspicture}
\end{center}

Conversely, it is widely believed that in every proper edge coloring of $K_n$, there is a
rainbow path on $n-1$ vertices (see for example~\cite{GyarfasMhalla2010}).
Still, this is far from proved, and to date, the best general lower bound on the number of vertices
in a maximum rainbow path in (a properly edge colored) $K_n$ is $(2n+1)/3$,
as proved by Gyárfás and Mhalla in~\cite{GyarfasMhalla2010}.
The main result of this thesis improves this bound to $(3/4-o(1))n$.
\begin{thm}
  \label{thm:intro}
  In every proper edge coloring of $K_n$, there is a rainbow path of length
  \[\left(\frac{3}{4}-o(1)\right)n\text. \]
\end{thm}

Several theorems and conjectures on rainbow cycles can be found in a paper by Akbari, Etesami, Mahini
and Mahmoody~\cite{Akbari2007}.
Most importantly (for our purposes),
it is proved that in every proper edge coloring of $K_n$, there is a rainbow cycle of length at least
$n/2-1$. This result was later improved on by Gyárfás, Ruszinkó, Sarközy and Schelp in~\cite{Gyarfas2011},
where a bound of $(4/7-o(1))n$ is given.

A related topic is that of colorful Hamiltonian cycles in proper edge colorings of $K_n$.  In~\cite{Akbari2007},
it is conjectured that every proper edge coloring of $K_n$ contains a Hamiltonian cycle using
at least $n-2$ colors, and it is proved that there is always one with at least $(2/3-o(1))n$
colors. The construction used in our proof of Theorem \ref{thm:intro} can be used to show
that there are Hamiltonian cycles using at least $(3/4-o(1))n$ colors.

In~\cite{Hahn1986}, Hahn and Thomassen studied
rainbow cycles and paths in \emph{$k$-bounded edge colorings} of $K_n$, that is,
(not necessarily proper) edge colorings that use every color at most $k$ times.
It is shown that for fixed $k$ and large enough $n$, every such coloring contains
a Hamiltonian rainbow path; and the authors conjecture that there are
Hamiltonian rainbow paths even if $k = a n$ for some suitably small constant factor $a$.

\section{Paths with repeated colors}

Generalizing the notion of a rainbow path, we consider paths in $K_n$
that use every color at most a constant number of times. We call such paths \emph{$k$-rainbow paths},
where $k$ is the number of times a color may appear on the path.
We will prove that in every proper edge coloring of $K_n$ and for every integer $k>0$, there are $k$-rainbow
paths on at least $n-O(n/k!)$ vertices. As far as we know, there are no previous results in this direction.

\section{Note on Latin squares}

We now give some motivation for the study of rainbow cycles and paths
by relating it to problems whose nature is not inherently graph-theoretic.

Latin squares have been a popular topic in combinatorics at least since
the times of Euler, who studied them extensively.
An array of $n$ rows and $n$ columns
is called a \emph{Latin square of order $n$} if every number in $[n] = \{1,\dotsc,n\}$
appears exactly once in each of its rows and columns.
A (complete) \emph{transversal} of a Latin square is a
selection of $n$ cells of the square, choosing exactly one from each row and column, such that every
number in $[n]$ is contained in exactly one cell of the selection.
Similarly, a \emph{partial transversal} of a Latin square is
a maximal selection of cells, each cell again being from a different row and column,
such that no two chosen cells contain the same symbol.

As shown by Maillet (1894), there are many Latin squares which do not have complete transversals.
However, a famous conjecture of Ryser (1967) states that every Latin square of odd order has
a transversal and, moreover, Brualdi conjectured that every Latin square of
order $n$ admits a partial transversal of size at least $n-1$.

Proofs for both conjectures seem out of reach, even though it is known
that in every Latin square of order $n$, there are partial transversals of size $n-o(n)$.
All this and more can be found in~\cite{Denes}.

How do Latin squares relate to rainbow cycles, or graph theory in general?
Consider any Latin square $A$, and let $R$ and $C$ denote the sets of its rows and columns respectively.
Then $A$ defines a proper edge coloring of the complete bipartite graph with partite sets
$C$ and $R$, as follows: for $c\in C$ and $r\in R$, the edge $\{c,r\}$ is
colored with the number contained in the cell determined
by the row $r$ and the column $c$. Then a transversal of $A$ corresponds to a perfect bipartite matching in this graph,
in which no two edges use the same color: a rainbow perfect matching.
Rainbow matchings are studied for example in~\cite{Wang2008}.

Conversely, every proper edge coloring of $K_n$ with
$n-1$ colors can be used to construct a Latin square $A$ of order $n$, in the following way.
Let $\{v_1,\dotsc,v_n\}$ denote the vertex set of $K_n$. Then for every $i\in [n]$ and $j\neq i$, let $A_{i,i}=n$
and let $A_{i,j}$ be the color of the edge $\{v_i,v_j\}$.
A complete transversal of this Latin square corresponds to
a 2-regular rainbow subgraph (i.e., a subgraph consisting of vertex-disjoint cycles)
that covers all but at most one vertex of $K_n$. This can be seen as follows. For every vertex $v_i$,
either $A_{i,i}$ is in the transversal, or two cells $A_{i,j}$ and $A_{k,i}$ (with $i\neq j$ and $i \neq k$) are in the transversal.
In the former case, $v_i$ does not belong to the subgraph. In the latter case, the two edges
$\{v_i,v_j\}$ and $\{v_i,v_k\}$ are included in the subgraph. Hence, all included vertices
have degree two.
By the defining property of
a transversal, all selected edges have different colors, so the subgraph is really rainbow,
and there are no cycles of length two. Moreover, for every vertex $v_i$, we have $A_{i,i}=n$,
so at most one vertex does not belong to the subgraph.

\section{Thesis structure}

In Chapter 2, some results on rainbow paths are proved; in particular, we prove Theorem
\ref{thm:intro}.
The chapter also contains an informal overview over the ideas used in the proof.

Chapter 3 takes a look at $k$-rainbow paths.
It is proved that every proper edge coloring of $K_n$ contains a $k$-rainbow path
on at least $(1-2/(k+2)!)n$ vertices.

The final chapter is a short conclusion.

\section{Notation}

In this section, we define the notation used throughout the thesis and briefly introduce
the basic graph-theoretic notions.

\subsection{Sets}

We write $\naturals$ for the set $\{1,2,3,\dotsc\}$ of natural numbers. If $n$ is a natural number, then
we write $[n]$ for the set $\{1,2,\dotsc,n\}$.

We use capital letters for sets.
If $A$ is a set, then $\abs{A}$ is the cardinality of $A$
and $\binom{A}{k}$ is the set of all $k$-element subsets of $A$.
We write $A^c$ for the complement of $A$ (relative to some universe).

\subsection{Graphs}

For graphs, we follow the notation from~\cite{Diestel}, although we will restate the most important
definitions here.

A \emph{graph} is a tuple $(V,E)$, where $V$ is the finite set of \emph{vertices}
and $E \subseteq \binom{V}{2}$ is the set of \emph{edges}.
The \emph{endpoints} of an edge are its elements, and the edge is \emph{incident}
to them and only them.
Two edges are \emph{coincident} if they intersect,
and two vertices $u$ and $v$ are \emph{adjacent} if $\{u,v\}\in E$.
If $G$ is a graph, then we write $V(G)$ for its vertex set and $E(G)$ for its edge set.
A graph $H=(V',E')$ is a \emph{subgraph} of $G=(V,E)$, written $H\subseteq G$, if $V'\subseteq V$ and $E'\subseteq E$.

For a graph $G=(V,E)$ and an arbitrary $e\in \binom{V}{2}$, we write
$G + e$ and $G-e$ for the graphs $(V,E \cup \{e\})$ and $(V,E\setminus \{e\})$.
If $G=(V,E)$ and $H=(V',E')$ are graphs,
then $G\cup H$ denotes the graph $(V\cup V', E\cup E')$.

The \emph{complete graph on $n$ vertices} is the graph with
vertex set $[n]$ and edge set $\binom{[n]}{2}$. It is denoted by $K_n$.

Given a graph $G = (V,E)$, a map $c\colon E \to \naturals$ is called a \emph{proper edge coloring}
(or simply a coloring) of $G$
if for every two coincident edges $e$ and $e'$ of $G$, we have $c(e)\neq c(e')$.
The colors in the image domain $c(E)$ of $c$ are called the colors \emph{used} by $G$, and
we usually write $c(G)$ for this set.
For an edge $\{u,v\}\in E$, we usually write $c(u,v)$ instead $c(\{u,v\})$. Slightly
abusing this notation, if $A$ is a set of vertices, then we also write $c(u,A)$ for
the set $c(E(u,A))$.

\subsection{Cycles and paths}

A \emph{path} is a non-empty graph $P=(V,E)$ of the form
\[V = \{p_1,p_2,\dotsc,p_k\} \quad \text{and}\quad E = \{\{p_1,p_2\},\{p_2,p_3\},\dotsc,\{p_{k-1},p_k\}\},\]
which we usually denote by the sequence $(p_1,p_2,\dotsc,p_k)$. Then $p_1$ and $p_k$ are the
\emph{start} and \emph{end} vertices of $P$, respectively. The number of edges in $E$ is called the \emph{length} of $P$. We call $p_i$ a \emph{$k$-successor} of $p_j$
if $i>j$ and there are at most $k-1$ vertices between $p_i$ and $p_j$ on $P$. In other words,
$p_i$ is a $k$-successor of $p_j$ if $0< i-j\leq k$. Equivalently,
$p_j$ is a \emph{$k$-predecessor} of $p_i$.

If $P = (p_1,p_2,\dotsc,p_k)$ is a path, then the graph $C=P + \{p_k,p_1\}$ is a \emph{cycle},
and $\abs{E(C)}$ is the \emph{length} of $C$. We represent this cycle by the cyclic sequence of
its vertices, for example $C = (p_1,p_2,\dotsc,p_k,p_1)$.
If $G$ is a graph and $H\subseteq G$ is a path or cycle such that $V(H)=V(G)$, then
$H$ is called \emph{Hamiltonian}.

\chapter{Rainbow Paths}

\section{Introduction}

Consider the complete graph $K_n =(V,E)$ with a proper edge coloring $c\colon E \to \naturals$.
Given this coloring, the rainbow paths and cycles are exactly the paths
and cycles in $K_n$ that use every color at most once.

If $P$ is a rainbow path, then we will refer to the colors in $c(P)$ as \emph{old}
and to those in $c(P)^c=c(E)\setminus c(P)$ as \emph{new}. Edges colored with
old colors are \emph{old}, edges colored with new colors are \emph{new}.

In~\cite{GyarfasMhalla2010}, Gyárfás and Mhalla proved that regardless of how the coloring
is chosen, there always are rainbow paths on at least $(2n+1)/3$ vertices.
Now we give the basic idea behind their proof, the details of which we will see later.
Consider any maximum rainbow path $P=(p_1,\dotsc,p_k)$ in $K_n$, and consider an edge
$\{p_i,p_{i+1}\}$ such that $\{p_1,p_{i+1}\}$ is new, as in the following
figure.
\begin{center}
  \begin{pspicture}(-5.2,-1.5)(5.2,2)
    \psset{fillstyle=solid,fillcolor=black}
    \cnode(-5,0){.08}{p1}  \nput{-90}{p1}{$p_1$}
    \cnode(0,0){.08}{pi}   \nput{-90}{pi}{$p_i$}
    \cnode(1,0){.08}{psi}   \nput{-90}{psi}{$p_{i+1}$}
    \cnode(5,0){.08}{pt}   \nput{-90}{pt}{$p_k$}

    \pnode(-2.75,0){a}
    \ncline{p1}{a}
    \pnode(-2.25,0){b}
    \ncline[nodesep=3pt,linestyle=dotted]{a}{b}
    \ncline{b}{pi}

    \pnode(2.75,0){a1}
    \ncline{psi}{a1}
    \pnode(3.25,0){b1}
    \ncline[nodesep=3pt,linestyle=dotted]{a1}{b1}
    \ncline{b1}{pt}

    \psset{fillstyle=none,arcangle=45}
    \ncarc{p1}{psi}
    \ncline[linewidth=1.75pt]{pi}{psi}
  \end{pspicture}
\end{center}
Clearly, any edge $\{p_k,r\}$ with $r\in V(P)^c$ cannot use the color
of $\{p_{i+1},p_i\}$, as otherwise the path
\[ (p_i,\dotsc,p_1,p_{i+1},\dotsc,p_k,r) \]
would be a rainbow path on $\abs{V(P)}+1$ vertices, contradicting
the choice of $P$. Viewed the other way around, we can say that a certain number
of edges in $E(P)$ are not allowed to have colors in $c(p_k,V(P)^c)$. But all
the edges in $E(p_k,V(P)^c)$ must be old, so their colors appear somewhere on the path.
As Gyárfás and Mhalla observed, this conflict leads to the bound $k\geq (2n+1)/3$.

But what if we knew that starting in any vertex $r\in V(P)^c$, there is a
rainbow path in $V(P)^c$ (of a certain minimum length $l$) that uses no colors of $c(P)$?
Assume that this is the case, and, moreover, that given two arbitrary new colors,
this path can be chosen in such a way
that it does not use any one of them. Then instead of forbidding colors
in $c(p_k,V(P)^c)$ to appear only on edges $\{p_i,p_{i+1}\}$ as above, we
can also forbid them to appear on any edge $\{p_i,p_{i+1}\}$ such that
$p_{i+1}$ has an $l$-successor $p_j$ with $c(p_1,p_j)\nin c(P)$.
While this does not immediately
lead to a good bound on the length of $P$, it does give us more flexibility; as we will see
we can now often simply `forget' about constant terms.
This is the main idea behind the upcoming
proof that in every proper edge coloring of $K_n$, there are rainbow paths of length $(3/4-o(1))n$.

Now for some definitions.
For any vertex $v\in V$,  we define the \emph{new neighborhood} of $v$ relative
to a rainbow path $P$ by
\[ \newn{v,P} = \{ u \in V\setminus \{v\} : c(u,v) \nin c(P) \}\text.\]
Analogously, if $C$ is a rainbow cycle, then
\[ \newn{v,C} = \{u \in V\setminus \{v\} : c(u,v) \nin c(C) \}\text. \]

Moreover, for any rainbow path $P=(p_1,\dotsc,p_k)$, we define the sets
\begin{align*} &A(P) = \{ p_i \in V(P) : p_{i+1}\in \newn{p_1,P}\}\\
\shortintertext{and}
&B(P) = \{ p_i \in V(P) : p_{i-1}\in \newn{p_{k},P}\}\text. \end{align*}
Note that these definitions are symmetric in the sense that if $P'=(p_k,\dotsc,p_1)$, then
$A(P)=B(P')$ and $A(P') = B(P)$. Figure \ref{fig:a} serves as a visual aid for
the formal definitions of $A(P)$ and $B(P)$.

\begin{figure}
  \centering
  \begin{pspicture}(-5.2,-2)(5.2,2)
    \psset{fillstyle=solid,fillcolor=black}
    \cnode(-5,0){.08}{p1}  \nput{-90}{p1}{$p_1$}
    \cnode(-3,0){.08}{x}
    \cnode(-2,0){.08}{y}
    \cnode(0,0){.08}{pi}
    \cnode(1,0){.08}{psi}
    \cnode(2,0){.08}{u}
    \cnode(3,0){.08}{v}
    \cnode(5,0){.08}{pk}   \nput{-90}{pk}{$p_k$}

    \pnode(-4.25,0){a}
    \ncline{p1}{a}
    \pnode(-3.75,0){b}
    \ncline[nodesep=3pt,linestyle=dotted]{a}{b}
    \ncline{b}{x}

    \pnode(-1.25,0){a1}
    \ncline{y}{a1}
    \pnode(-.75,0){b1}
    \ncline[nodesep=3pt,linestyle=dotted]{a1}{b1}
    \ncline{b1}{pi}

    \pnode(3.75,0){a2}
    \ncline{v}{a2}
    \pnode(4.25,0){b2}
    \ncline[nodesep=3pt,linestyle=dotted]{a2}{b2}
    \ncline{b2}{pk}

    \psset{fillstyle=none,arcangle=45}
    \ncarc{p1}{psi}
    \ncline{x}{y}
    \ncline{pi}{psi}
    \ncline{psi}{u}
    \ncline{u}{v}
    \ncarc{p1}{y}
    \ncarc{u}{pk}
    \pscircle(-3,0){.25}
    \pscircle(0,0){.25}
    \psframe(2.75,-.25)(3.25,.25)
  \end{pspicture}
  \caption{Vertices in $A(P)$ (circled) and $B(P)$ (framed)}
  \label{fig:a}
\end{figure}

We will be working mostly with maximum rainbow paths, that is,
rainbow paths of maximum length. Clearly, if $P = (p_1,\dotsc,p_k)$ is a maximum
rainbow path, then adding any edge from $E(p_k,V(P)^c)$ to it cannot result
in a rainbow path (otherwise $P$ would not be maximum).
This means that all edges in $E(p_k,V(P)^c)$ use colors that are also used by $P$.
The same argument can be made for edges in $E(p_1,V(P)^c)$.

\begin{proposition}
  \label{prop:maximality}
  If $P=(p_1,\dotsc,p_k)$ is a maximum rainbow path with respect to some coloring $c$,
  then $c(p_1,V(P)^c)\subseteq c(P)$ and $c(p_k,V(P)^c)\subseteq c(P)$.

  In particular, $\abs{A(P)} = \abs{\newn{p_1,P}} \geq n-k$ and $\abs{B(P)} = \abs{\newn{p_k,P}} \geq n-k$.
\end{proposition}

Using this simple observation, we directly obtain the following lower bound on the length of
a maximum rainbow path.

\begin{proposition}
  \label{prop:n/2-bound}
  In every proper edge coloring of $K_n$, there are rainbow paths on at least
  $(n+1)/2$ vertices.
\end{proposition}
\begin{proof}
  Consider an arbitrary proper edge coloring $c$ of $K_n$,
  and let $P = (p_1,\dotsc,p_k)$ be a maximum rainbow path in this coloring.
  We have
  \[ \abs{c(p_k,V(P)^c)} = \abs{V(P)^c} = n-k\text, \]
  and
  \[ \abs{c(P)} = \abs{E(P)} = k-1\text. \]

  By Proposition \ref{prop:maximality}, we have $c(p_k,V(P)^c)\subseteq c(P)$, and hence
  $\abs{c(p_k,V(P)^c)} \leq \abs{c(P)}$. Thus
  $n-k \leq k-1$,
  and so we have
  \[ k \geq \frac{n+1}{2}\text, \]
  as claimed.
\end{proof}

Actually, we proved something slightly stronger, namely, that given any vertex $v$ of $K_n$,
there is a rainbow path on $(n+1)/2$ vertices starting in this vertex. To see this,
revisit the proof and let $P$ be the longest rainbow path starting in $v$,
and observe that the argument made for Proposition \ref{prop:maximality} still applies.

\section{Rotations}

The paths in $K_n$ admit many symmetries; in fact, every
permutation of the vertices of a path results in another path.
We consider a special kind of permutation, which we call rotation here. Rotations were
already used by Pósa in~\cite{Posa1976}.
For every $i\in[k]$, there is a rotation $\rho_i$
which acts on the path $P = (p_1,\dotsc,p_k)$ to produce
the path \[\rho_i \cdot P = (p_i,p_{i-1},\dotsc,p_1,p_{i+1},p_{i+2},\dotsc,p_k)\text,\]
as shown in figure \ref{fig:rotation}.

\begin{figure}
  \centering
  \begin{pspicture}(-5.2,-2)(5.2,2)
    \psset{fillstyle=solid,fillcolor=black}
    \cnode(-5,0){.08}{p1}  \nput{-90}{p1}{$p_1$}
    \cnode(0,0){.08}{pi}   \nput{-90}{pi}{$p_i$}
    \cnode(1,0){.08}{psi}   \nput{-90}{psi}{$p_{i+1}$}
    \cnode(5,0){.08}{pt}   \nput{-90}{pt}{$p_k$}

    \pnode(-2.75,0){a}
    \ncline{p1}{a}
    \pnode(-2.25,0){b}
    \ncline[nodesep=3pt,linestyle=dotted]{a}{b}
    \ncline{b}{pi}

    \pnode(2.75,0){a1}
    \ncline{psi}{a1}
    \pnode(3.25,0){b1}
    \ncline[nodesep=3pt,linestyle=dotted]{a1}{b1}
    \ncline{b1}{pt}

    \psset{fillstyle=none,arcangle=45}
    \ncarc{p1}{psi}
  \end{pspicture}
  \caption{The path $\rho_i\cdot (p_1,\dotsc,p_k)$}
  \label{fig:rotation}
\end{figure}

The point is that
if $P = (p_1,\dotsc,p_k)$ is a rainbow path and
$p_i\in A(P)$, then
$\rho_i\cdot P$ is a rainbow path that does not use the color
$c(p_i,p_{i+1})$, but that is still very similar to $P$. In particular,
$\rho_i\cdot P$ ends in the same vertex as $P$.

\begin{proposition}
  \label{prop:maximality-rot}
  If $P = (p_1,\dotsc,p_k)$ is a maximum rainbow path with respect to some coloring $c$, then
  for every $p_i\in A(P)$ we have $c(p_i,p_{i+1})\nin c(p_k,V(P)^c)$. Similarly,
  for every $p_i\in B(P)$ we have $c(p_i,p_{i-1}) \nin c(p_1,V(P)^c)$.
\end{proposition}
\begin{proof}
  Suppose that $p_i\in A(P)$ is such that $c(p_i,p_{i+1})\in c(p_k,V(P)^c)$.
  Then there is a vertex $v\in V(P)^c$ such that $c(p_k,v)=c(p_i,p_{i+1})\nin  c(\rho_i\cdot P)$.
  Since $\rho_i\cdot P$ ends in $p_k$ and $v\nin V(P) = V(\rho_i\cdot P)$,
  the rainbow path $\rho_i\cdot P$ violates Proposition
  \ref{prop:maximality}.

  The second part follows by symmetry.
\end{proof}

This fact was used by
by Gyárfás and Mhalla to find rainbow paths on $(2n+1)/3$ vertices. Because
the proof is quite elegant and uses a technique similar to those used later on, we
give it here.

\begin{thm}[{\cite{GyarfasMhalla2010}}]
  \label{thm:gm2010}
  In every proper edge coloring of $K_n$, there is a rainbow path on at least
  $(2n+1)/3$ vertices.
\end{thm}
\begin{proof}
  Consider an arbitrary proper edge coloring $c$ of $K_n$,
  and let $P = (p_1,\dotsc,p_k)$ be a maximum rainbow path with respect to $c$.
  By Proposition \ref{prop:maximality},
  $c(p_k,V(P)^c)\subseteq c(P)$.
  Now let
  \[ X = \{ c(p_i,p_{i+1}) : p_i\in A(P) \}\text. \]
  Then we have $\abs{X} = \abs{A(P)} \geq n-k$.
  Since $X\subseteq c(P)$, we get \[X\cup c(p_k,V(P)^c)\subseteq c(P)\text,\] and hence
  $\abs{X \cup c(p_k,V(P)^c)} \leq \abs{c(P)} = \abs{E(P)} = k-1$.
  By Proposition \ref{prop:maximality-rot}, we have
  $\abs{X \cap c(p_k,V(P)^c)} = \emptyset$, and thus
  \[ \abs{X\cup c(p_k,V(P)^c)} = \abs{X} + \abs{c(p_k,V(P)^c)} \geq n-k + n-k\text,\] and so
  $2n - 2k\leq k-1$, or
  \[ k \geq \frac{2n+1}{3}\text, \]
  which is what we needed to prove.
\end{proof}

Now we come to the main result of this work.

\section{A theorem on the length of rainbow paths}

In this section, we are going to prove the following result, which is equivalent to Theorem
\ref{thm:intro}.

\begin{thm}
  \label{thm:awesome}
  For every $\epsilon > 0$, there is some $n_0 = n_0(\epsilon)$ such that
  for $n> n_0$, every proper edge coloring of $K_n$ contains a rainbow path on at least
   $(3/4-\epsilon)n$ vertices.
\end{thm}

The proof proceeds by contradiction, that is, we assume that for some
$\epsilon > 0$, there are no such rainbow paths, and show that we can construct a
rainbow path that is longer than the length
of a (supposedly) maximum rainbow path. There are two steps to how this is done:
\begin{enumerate}
  \parskip 0pt
\item We show that there is a maximum rainbow path $P$ such that almost all (i.e., all but constantly many)
  vertices in $V(P)^c$ have at least some constant number of new neighbors in $V(P)^c$.
\item We show how such a rainbow path can be extended to a longer rainbow path, resulting in a contradiction.
\end{enumerate}

\subsection{Preliminaries}

Let $\epsilon > 0$ and let $n>n_0$ for some suitably large value $n_0 = n_0(\epsilon)$.
We will not give an explicit value for $n_0$, rather we will tacitly assume that
increasing functions of $n$ dominate any constant.
We are given a proper edge coloring $c$ of $K_n$.
Then we denote by $t$ the number of vertices in a maximum length rainbow path, that is,
there are no rainbow paths of length $t$ (on $t+1$ vertices).

We will assume throughout that we have
\[ t \leq \left(\frac{3}{4}-\epsilon\right)n\text,\]
trying to arrive at a contradiction.

During the proof, let $a$ be a `large enough' constant. For example $a = 100+\ceil{100/\epsilon}$
is more than enough.

\subsection{Nice rainbow paths}

For any rainbow path $P$, we define the set
 \[ R(P) = \{ r\in V(P)^c : \abs{\newn{r,P} \cap V(P)^c} > a\}  \text. \]
So $R(P)$ is the set of vertices in $V(P)^c$ that have more than $a$ new neighbors
in $V(P)^c$. Then we have the following definition.

\begin{define}[Nice rainbow path]
  A \emph{nice} rainbow path is a rainbow path $P$ satisfying
  \[ \abs{R(P)}> n-t- 1/\epsilon\text. \]
\end{define}
In other words, nice rainbow paths are such that all but at most $1/\epsilon$ vertices in $V(P)^c$
have more than $a$ new neighbors in $V(P)^c$.
We will be interested in nice maximum rainbow paths, that is, nice rainbow paths on $t$ vertices.
Before showing that such paths exist, we will motivate them by proving that they have some nice properties.

\begin{proposition}
  \label{prop:long-paths}
  If $P$ is nice maximum rainbow path,
  then for any vertex $r\in R(P)$ and any set $F$ of colors, there is a rainbow path $Q$
  starting in $r$ in the subgraph induced by $R(P)$ such that
  $\abs{V(Q)} = \floor{(a-\abs{F}-1/\epsilon)/2}$
  and $c(Q)\cap (F \cup c(P)) = \emptyset$.
\end{proposition}
\begin{proof}
  Let $k=\floor{(a-\abs{F}-1/\epsilon)/2}$.
  We construct a series of rainbow paths $Q_1,\dotsc,Q_{k}$ starting in $r$
  and using only vertices in $R(P)$, such that for every $i \in [k]$, we have $\abs{V(Q_i)} = i$. Furthermore,
  we shall make sure that no $Q_i$ uses colors in $F\cup c(P)$. Clearly, the path $Q_k$
  will have the desired properties.

  Let $Q_1 = (r)$ and define $Q_{i+1}$ in terms of $Q_i$ as follows.
  By construction, $Q_i$ starts in $r$ and ends in some vertex $v\in R(P)$.
  If there is a vertex $u\in \newn{v,P} \cap R(P)\setminus V(Q_i)$ such that
  \[ c(u,v)\nin c(Q_i)\cup F\text, \]
  then we can simply define $Q_{i+1}$ to be $Q_i\cup(v,u)$.
  Hence it is enough to prove that $\abs{\newn{v,P} \cap R(P) \setminus V(Q_i)} > \abs{c(Q_i)\cup F}$.

  Because $i< k$, we have
  \[ \abs{c(Q_i)\cup F} \leq \abs{c(Q_i)} + \abs{F} < k + \abs{F}\text. \]
  Since $P$ is nice, we have $\abs{V(P)^c\setminus R(P)} < 1/\epsilon$, and as $v\in R(P)$, we get
  \[ \abs{\newn{v,P} \cap R(P)} > a-1/\epsilon\text, \]
  and hence
  \[ \abs{\newn{v,P} \cap R(P) \setminus V(Q_i)} > a- 1/\epsilon -k\text. \]
  Now the claim follows from
  \[ a-1/\epsilon -k \geq k + \abs{F}\text, \]
  which holds because
  $k \leq (a-\abs{F}-1/\epsilon)/2$.
\end{proof}

The importance of this fact comes from the following lemma.

\begin{lemma}
  \label{lemma:no-paths}
  If $P$ is a nice maximum rainbow path, then
  every rainbow path \[Q = (q_1,\dotsc,q_k)\] satisfies at least one of the following properties:
  \begin{enumerate}[(P1)]
    \parskip 0pt
    \item $V(Q)\not\subseteq V(P)$,
    \item $(\newn{q_1,Q}\cup \newn{q_k,Q}) \cap R(P)= \emptyset$,
    \item $\abs{c(Q)\setminus c(P)}> 2$, or
    \item $k< t-a/3$.
  \end{enumerate}
\end{lemma}
\begin{proof}
  Suppose that there is a rainbow path $Q=(q_1,\dotsc,q_k)$ satisfying none of the properties (P1--4).

  Because (P2) is violated, either $\newn{q_1,Q}\cap R(P)\neq \emptyset$ or
  $\newn{q_k,Q}\cap R(P)\neq \emptyset$.
  In any case, there is an endpoint $q$ of $Q$ such that
  for some vertex $s_1\in R(P)$, we have $c(q,s_1)\nin c(Q)$.
  Then let \[F = \{c(q,s_1)\}\cup (c(Q)\setminus c(P))\text,\]
  and since (P3) is violated, we have $\abs{F}\leq 3$.

  By Proposition \ref{prop:long-paths}, we know that the subgraph induced by $R(P)$ contains
  a rainbow path $S = (s_1,\dotsc,s_l)$ with $l\geq \floor{(a-3-1/\epsilon)/2}$
  such that $c(S)\cap (F\cup c(P)) = \emptyset$. So $S$ does not use any colors in
  \[ F\cup c(P) = \{c(q,s_1)\}\cup c(Q)\cup c(P)\text, \]
  and hence the path $S' = (q_1,s_1,\dotsc,s_l)$ is a rainbow path with $c(S')\cap c(Q)=\emptyset$.

  Because (P1) is violated, we have $V(Q)\cap R(P) = \emptyset$, and thus
  \[ V(Q)\cap V(S') = \{q\}\text.\]
  so $Q\cup S'$ is a rainbow path. As (P4) is violated, we have
  \begin{align*}
    \abs{V(Q\cup S')} &=\abs{V(Q)}+\abs{V(S')}-1\\
    &= k + l-1\\
    &\geq t-\frac{a}{3} + \floor{(a-3-1/\epsilon)/2}-1\\
    &> t-\frac{a}{3} + \frac{a}{2} -\frac{1}{2\epsilon}- 3\\
    &> t\text.
  \end{align*}
  But by definition of $t$, there are no rainbow paths this long.
\end{proof}

Thus Lemma \ref{lemma:no-paths} suggests that we can prove Theorem \ref{thm:awesome} in two steps. Under the assumption
that $t\leq(3/4-\epsilon)n$:
\begin{itemize}
\item Show that there is at least one nice maximum rainbow path $P$.
\item Then show that for this $P$, there is a rainbow path $Q$ violating Lemma \ref{lemma:no-paths}.
\end{itemize}

\subsection{Existence of nice maximum rainbow paths}

In this section, we will prove that there are nice maximum rainbow paths.
First, let us reiterate what it means for a rainbow path not to be nice. If $P$
is a rainbow path that is not nice, then we have
\[ \abs{R(P)}\leq n-t-1/\epsilon\text, \]
or, equivalently,
\[ \abs{R(P)^c} \geq n-(n-t-1/\epsilon) = t+1/\epsilon\text. \]
For any rainbow path, we have $\abs{V(P)}\leq t$, so we get
\[ \abs{R(P)^c\cap V(P)^c} = \abs{R(P)^c\setminus V(P)}\geq 1/\epsilon\text. \]
This means that there are at least $1/\epsilon$ vertices in $V(P)^c$
that have no more than $a$ new neighbors in $V(P)^c$.

\begin{proposition}
  \label{cons:no-cycles}
  Suppose that there are no nice maximum rainbow paths.
  Then there are no rainbow cycles of length $t$.
\end{proposition}
\begin{proof}
  Suppose that $C$ is such a cycle. Then removing any edge $e\in E(C)$ from $C$,
  we get a maximum rainbow path $P=C-e$, with $V(P)=V(C)$.

  By assumption, this maximum rainbow path cannot be not nice,
  so there is a vertex $v \in V(C)^c$ such that at least $\abs{\newn{v,P}} - a\geq n-t-a$
  new neighbors of $v$ are in $V(C)$.
  Since, for large enough $n$, we have
  $n-t-a\geq 2$, at
  least one of those neighbors, call it $u$, is such that $c(u,v)\neq c(e)$.
  So $c(u,v)\nin c(P+e) = c(C)$.

  Let $u'$ be a vertex adjacent to $u$ on $C$. Then
  $(C-\{u,u'\}) \cup (u,v)$ is a rainbow path of length $t$. This contradicts the definition of $t$.
\end{proof}

\begin{proposition}
  Suppose that there are no nice maximum rainbow paths, and let
  $P = (p_1,\dotsc,p_t)$ be a maximum rainbow path. Then for every vertex
  $p_i\in A(P)\cap B(P)$ we have $c(p_1,p_{i+1}) = c(p_{i-1},p_t)$.
\end{proposition}
\begin{proof}
  Assume by way of contradiction
  that there is a vertex $p_i \in A(P)\cap B(P)$ such that $c(p_1,p_{i+1}) \neq c(p_{i-1},p_t)$.
  By definition, $c(p_1,p_{i+1})\nin c(P)$ and $c(p_t,p_{i-1})\nin c(P)$,
  so the cycle
  \[ C = (p_1,\dotsc,p_{i-1},p_t,p_{t-1},\dotsc,p_{i+1},p_1) \]
  shown in figure \ref{fig:easy-cycle} is a rainbow cycle of length $t-1$.

  \begin{figure}
    \centering
    \begin{pspicture}(-5.2,-2)(5.2,2)
      \psset{fillstyle=solid,fillcolor=black}
      \cnode(-5,0){.08}{x1}  \nput{-90}{x1}{$p_1$}
      \cnode(-1,0){.08}{xpi} \nput{-90}{xpi}{$p_{i-1}$}
      \cnode(0,0){.08}{xi}   \nput{-90}{xi}{$p_i$}
      \cnode(1,0){.08}{xsi}   \nput{-60}{xsi}{$p_{i+1}$}
      \cnode(5,0){.08}{xt}   \nput{-90}{xt}{$p_t$}

      \pnode(-3.25,0){a}
      \ncline{x1}{a}
      \pnode(-2.75,0){b}
      \ncline[nodesep=3pt,linestyle=dotted]{a}{b}
      \ncline{b}{xpi}

      \pnode(2.75,0){a1}
      \ncline{xsi}{a1}
      \pnode(3.25,0){b1}
      \ncline[nodesep=3pt,linestyle=dotted]{a1}{b1}
      \ncline{b1}{xt}

      \ncline[linestyle=dotted,nodesep=3pt]{xa}{xpi}

      \psset{fillstyle=none,arcangle=45}
      \ncarc{xsi}{x1}
      \ncarc{xa}{xt}
      \ncarc{xpi}{xt}
    \end{pspicture}
    \caption{A rainbow cycle if $p_i\in A(P)\cap B(P)$ and $c(p_1,p_{i+1})\neq c(p_{i-1},p_t)$}
    \label{fig:easy-cycle}
  \end{figure}

  By assumption, $P$ could not have been nice, so there is a vertex
  $v \in V(P)^c$ with
  \[ \abs{\newn{v,P}\cap V(P)} \geq \abs{\newn{v,P}} - a\text. \]
  Some thought shows that
  \[ \abs{\newn{v,C} \cap V(C)}\geq \abs{\newn{v,C}}-a-5\text, \]
  because at most three vertices of $\newn{v,P}\cap V(P)$ are in not in
  $\newn{v,C}\cap V(C)$, and at most two vertices of $\newn{v,C}$ are
  not in $\newn{v,P}$.

  But no two vertices $u,u'\in \newn{v,C}\cap V(C)$ can be adjacent in $C$, because otherwise,
  the cycle $(C - \{u,u'\}) \cup (u,v, u')$ would be
  a rainbow cycle of length $t$, contradicting Proposition \ref{cons:no-cycles}.
  In other words, no edge $e\in E(C)$ is incident to two vertices in  $\newn{v,C} \cap V(C)$. Hence,
  if we write $X = \{e \in E(C) : e\cap \newn{v,C} \neq \emptyset\}$, then
  \begin{align*}
    \abs{X} &= 2\abs{\newn{v,C}\cap V(C)}\\
    &\geq 2\abs{\newn{v,C}}-2a-10\\
    &\geq 2n-2t-2a-12\text.
  \end{align*}

  Now let $Y = V(C)^c \cap \newn{v,C}^c$. We have
  \begin{align*}
    \abs{Y} &= \abs{V(C)^c \cap \newn{v,C}^c}\\
    &= \abs{V(C)^c} - \abs{\newn{v,C} \cap V(C)^c}\\
    &= n-t+1 - \abs{\newn{v,C}} + \abs{\newn{v,C}\cap V(C)}\\
    &\geq n-t+1 - \abs{\newn{v,C}} + \abs{\newn{v,C}}-a-5\\
    &= n-t-a-4\text.
  \end{align*}
  Clearly, $c(X) \subseteq c(C)$ and $c(v,Y\setminus \{v\})\subseteq c(C)$.
  Then $c(X)$ and $c(v,Y\setminus \{v\})$ have to intersect, because
  \begin{align*}
    \abs{c(X)} + \abs{c(v,Y\setminus \{v\})} &= \abs{X} + \abs{Y} - 1\\
    &\geq 3n-3t - 3a -17 \\
    &> t-1\\
    &= \abs{c(C)}\text.
  \end{align*}
  Therefore, there exists an edge $e\in X$ with $c(e) = c(v,u)$ for some $u\in V(C)^c$.
  By definition, $e$ is incident to some vertex $x\in V(C)$ with $c(x,v)\nin c(C)$. Then
  the rainbow path $(C-e)\cup (x,v,u)$ is a rainbow path of length $t$, but this contradicts the definition
  of $t$.
\end{proof}

This suggests that one could choose $P$ in such a way that
$\abs{A(P)\cap B(P)}$ is small. Indeed, we can do this, as we will
show next.

\begin{proposition}
  \label{cons:a-b}
  Suppose that there are no nice maximum rainbow paths.
  Then there is a maximum rainbow path $P$ such that $\abs{A(P)\cap B(P)}\leq \epsilon n$.
\end{proposition}
\begin{proof}
  Let $k = \ceil{1+2/\epsilon}$ and let $P$ be an arbitrary maximum rainbow path.
  Now we define a sequence of maximum rainbow paths $(P_0,\dotsc,P_k)$ as follows:
  \begin{itemize}
    \parskip 0pt
  \item $P_0 = P$.
  \item $P_l$ is defined in terms of $P_{l-1}=(p_1,\dotsc,p_t)$. Because $P_{l-1}$ is a maximum rainbow path,
    we have $\abs{A(P_{l-1})} \geq n-t > k$. Then choose a vertex $p_i \in A(P_{l-1})$ that
    is \emph{not} the starting vertex of any of the paths $P_0,\dotsc,P_{l-1}$,
    and let $P_l = \rho_i \cdot P_{l-1}$. Then, because $p_i\in A(P_{l-1})$, the path $P_l$
    is a maximum rainbow path starting in $p_i$ and ending in $p_t$.
  \end{itemize}

  By this construction, all paths $P_i$ end in the same vertex $p_t$, but no two different
  paths start in the same vertex.
  We now show that at least one of the paths $P_l$ must
  satisfy $\abs{A(P_l)\cap B(P_l)}\leq \epsilon n$. By way of contradiction,
  assume that for every $l\in \{0,\dotsc,k\}$, we have
  \[ \abs{A(P_l)\cap B(P_l)} > \epsilon n\text.\]
  For every path $P_l = (p_1,\dotsc,p_t)$, we define two sets $X_l$ and $Y_l$.

  Let $X_l$ be the smallest set
  that, for every vertex $p_i\in \newn{p_t,P_l}$, contains the triples $(p_i,p_{i+1},p_{i+2})$ and
  $(p_i,p_{i-1},p_{i-2})$, if those vertices are defined.
  Observe that $\abs{X_l\setminus X_{l-1}} \leq 4$, because the path $P_l$ is built from $P_{l-1}$ by a single rotation.
  Then, since $\abs{X_0} \leq 2t$, we get
  \[ \abs{X_0\cup \dotsb\cup X_k} \leq 2t + 4k\text. \]
  Now let $Y_l$ be the smallest set that, for every vertex $p_i \in A(P_l)\cap B(P_l)$, contains the
  triple $(p_{i-1},p_i,p_{i+1})$. If $p_i\in B(P_l)$, then $p_{i-1}\in \newn{p_t,P_l}$, so we have $Y_l\subseteq X_l$.

  Assuming that $n$ is large enough, we get
  \[ \frac{\abs{X_0\cup \dotsb\cup X_k}}{\abs{Y_l}} < \frac{2t+4k}{\epsilon n} < 1+2/\epsilon\leq k\text.\]
  Thus $k\abs{Y_l}> \abs{X_0\cup \dotsb \cup X_k}$. Since $Y_l\subseteq \abs{X_0\cup \dotsb \cup X_k}$,
  and because there are $k$ sets $Y_l$ altogether, this means that
  there are
  two sets $Y_i$ and $Y_j$ that intersect, say in the triple $(u,v,w)$.
  But if we write $P_i = (p_1,\dotsc,p_t)$ and $P_j=(p_1',\dotsc,p_{t-1}',p_t)$, then either
  $c(p_1,w)\neq c(u,p_t)$ or $c(p_1',w)\neq c(u,p_t)$, contradicting the result above stating that
  for any maximum rainbow path $P=(p_1,\dotsc,p_t)$, every vertex $p_i\in A(P)\cap B(P)$
  satisfies $c(p_1,p_{i+1}) = c(p_{i-1},p_t)$.
\end{proof}

Now we are ready to prove the following.

\begin{lemma}
  There exists at least one nice maximum rainbow path.
\end{lemma}
\begin{proof}
  We may assume that there are no nice maximum rainbow paths.

  Then, by Proposition \ref{cons:a-b}, there is a maximum rainbow path $P$ with
  $\abs{A(P)\cap B(P)}\leq \epsilon n$. Writing $A$ and $B$ for
  $A(P)$ and $B(P)$ respectively, this means that
  \[ \abs{A \cup B} = \abs{A} + \abs{B} - \abs{A\cap B} \geq 2n-2t - \epsilon n\text. \]
  Observe that every vertex in $A\cup B$ can have at most one new neighbor in $V(P)^c$,
  as otherwise we would get a rainbow path of length $t$.

  Now suppose that $r$ is a vertex in $V(P)^c$ such that
  \[ \abs{\newn{r,P}\cap V(P)} \geq \newn{r,P}-a\geq n-t-a\text. \]
  For brevity, let us write $X = \newn{r,P}\cap V(P)$.
  Then, using $\abs{A \cup B \cup X}\leq t$,
  \begin{align*}
    \abs{(A \cup B) \cap X}& = \abs{A\cup B} + \abs{X} - \abs{A \cup B \cup X}\\
    & \geq 2n-2t-\epsilon n + n - t-a - t\\
    & = 3n - 4t -a- \epsilon n\\
    & \geq \epsilon n\text.
  \end{align*}

  Therefore, at least $\epsilon n$ new neighbors of $r$
  are in $A\cup B$.
  Hence, by the observation above, there may be at most
  \[ \frac{\abs{A\cup B}}{\epsilon n} \leq \frac{2n-2t-\epsilon n}{\epsilon n} < \frac{1}{\epsilon} \]
  vertices $r\in V(P)^c$ with $\abs{\newn{r,P} \cap V(P)} \geq \abs{\newn{r,P}}- a$. In
  other words, more than $n-t-1/\epsilon$ vertices in $V(P)^c$ satisfy
  \[ \abs{\newn{r,P}\cap V(P)}> a\text, \]
  so $P$ is nice, contradicting our assumption.
\end{proof}

\subsection{Establishing a contradiction}

In the following, let $P = (p_1,\dotsc,p_t)$ be a nice maximum rainbow path, that is, such that
\[ \abs{R(P)} > n-t-1/\epsilon\text.\]
We will now try to reach a contradiction by constructing a rainbow path $Q$ violating
Lemma \ref{lemma:no-paths}. We will start by proving the following useful proposition.

\begin{proposition}
  \label{prop:counting}
  Let $A$ be any subset of $V(P)$. Then for any $k \in \naturals$,
  there are at most $1+t/k$ vertices in $A$ that do not have a $k$-successor in $A$ (on $P$).
\end{proposition}
\begin{proof}
  Let $X\subseteq A$ be the set of vertices in $A$ that do not have a $k$-successor in $A$.
  Clearly, there can be only one vertex in $A$ that does not have a $t$-successor in $A$.
  So for all but at most one vertex in $X$, there are $k$ vertices following that vertex that are not in $A$.
  If we also count the vertices in $X$ themselves, then in total we count at least
  $k(\abs{X}-1)+\abs{X}$ vertices. Altogether, there are only $t$ vertices in $P$. Therefore
  \[ k(\abs{X}-1)+\abs{X} \leq t \]
  and hence
  \[ \abs{X} \leq \frac{t+k}{k+1} \leq \frac{t}{k}+1\text. \]
  This completes the proof.
\end{proof}

\begin{proposition}
  \label{prop:a-b-bound}
  We have $\abs{A(P) \cap B(P)} \leq \epsilon n$.
\end{proposition}
\begin{proof}
  Suppose that \[\abs{A(P) \cap B(P)} >  \epsilon n\text.\]
  Let $k=\ceil{t/(\epsilon n-1)}$.
  If $p_i \in A(P) \cap B(P)$, then by definition $p_{i-1} \in \newn{p_t,P}$.
  By Proposition \ref{prop:counting}, at most $1+t/k$ vertices in $\newn{p_t,P}$ have no $k$-predecessor in $\newn{p_t,P}$
  on $P$.
  Therefore, there are more than
  \[ \epsilon n - 1-t/k \geq 0\]
  vertices $p_i\in A(P)\cap B(P)$ such that $p_{i-1}$ does have a $k$-predecessor $p_j$ in $\newn{p_t,P}$.
  Take any such vertex;
  at least one of the edges $\{p_t,p_{i-1}\}$ and $\{p_t,p_{j}\}$ is colored differently from $\{p_1,p_{i+1}\}$.
  As shown in figure \ref{fig:cycle-ab}, there is a rainbow cycle $C$ using this edge, of length
  at least \[t-k-1 \geq t-\ceil{t/(\epsilon n-1)}-1\geq t- a/3\text.\] Now we distinguish two cases.

  \begin{figure}
    \centering
    \begin{pspicture}(-5.2,-2)(5.2,2)
      \psset{fillstyle=solid,fillcolor=black}
      \cnode(-5,0){.08}{x1}  \nput{-90}{x1}{$p_1$}
      \cnode(-2,0){.08}{xa}  \nput{-90}{xa}{$p_j$}
      \cnode(-1,0){.08}{xpi} \nput{-90}{xpi}{$p_{i-1}$}
      \cnode(0,0){.08}{xi}   \nput{-90}{xi}{$p_i$}
      \cnode(1,0){.08}{xsi}  \nput{-60}{xsi}{$p_{i+1}$}
      \cnode(5,0){.08}{xt}   \nput{-90}{xt}{$p_t$}

      \pnode(-3.75,0){a}
      \ncline{x1}{a}
      \pnode(-3.25,0){b}
      \ncline[nodesep=3pt,linestyle=dotted]{a}{b}
      \ncline{b}{xa}

      \pnode(2.75,0){a1}
      \ncline{xsi}{a1}
      \pnode(3.25,0){b1}
      \ncline[nodesep=3pt,linestyle=dotted]{a1}{b1}
      \ncline{b1}{xt}

      \ncline[linestyle=dotted,nodesep=3pt]{xa}{xpi}

      \psset{fillstyle=none,arcangle=45}
      \ncarc{xsi}{x1}
      \ncarc{xa}{xt}
      \ncarc{xpi}{xt}
    \end{pspicture}
    \caption{Building a rainbow cycle for $p_i \in A(P)\cap B(P)$}
    \label{fig:cycle-ab}
  \end{figure}

  \begin{enumerate}[{Case }1.]
    \parskip 0pt
  \item There is a vertex $r\in R(P)$ such that for some $v \in V(C)$, we have $c(r,v)\nin c(C)$.
    Let $e\in E(C)$ be adjacent to $v$. Then it is easily verified that
    the path $C-e$ violates Lemma \ref{lemma:no-paths}.

  \item No vertex $r\in R(P)$ has a neighbor $v\in V(C)$ such that $c(r,v) \nin c(C)$.
    Let $r$ be any vertex in $R(P)$. Since there are at most $\abs{c(C)}$
    vertices $u$ such that $c(r,u) \in c(C)$, and because by assumption all of them are in $V(C)$,
    we have $c(r,p_i)\nin c(C)$.
    Then let $e\in E(C)$ be adjacent to $p_{i+1}$; in this case the path $(C-e)\cup (p_{i+1},p_i)$
    violates Lemma \ref{lemma:no-paths}.
  \end{enumerate}
  Both cases result in a contradiction.
\end{proof}

\begin{proposition}
  There are at least $\epsilon n$ vertices $p_i \in A(P)$ such that
  $c(p_i,p_{i+1}) \in c(p_1,R(P))$.
\end{proposition}
\begin{proof}
  Let \[X = \{ p_i \in V(P) : c(p_{i},p_{i+1}) \in c(p_1,R(P)) \} \text.\]
  First we show that $X\cap B(P) = \emptyset$. This is the case, because if
  $p_i \in X\cap B(P)$, then  $c(p_i,p_{i+1}) = c(p_1,r)$ for some $r\in R(P)$, so
  $P' = (r,p_1,\dotsc,p_{i-1},p_t,p_{t-1},\dotsc,p_{i+1})$ is a rainbow path on $t$ vertices.
  But since $r\in R(P)$, at least one vertex of $\newn{r,P'}$ is not in $V(P')$, so
  $P'$ cannot be maximum.

  Hence $X\cap B(P) = \emptyset$.
  Then we can partition $X$ as follows:
  \[ X = \big(X \cap (A(P)\cup B(P))^c\big) \cup \big(X\cap A(P)\big)\text. \]
  By Proposition \ref{prop:a-b-bound}, we have $\abs{A(P)\cap B(P)}\leq \epsilon n$.
  So we get
  \begin{align*}
    \abs{X \cap (A(P)\cup B(P))^c} &\leq  \abs{V(P)} - \abs{A(P)\cup B(P)}\\
    &=  \abs{V(P)} - \abs{A(P)} -\abs{B(P)} + \abs{A(P)\cap B(P)}\\
    &\leq t - n+t - n+t + \epsilon n\\
    & = 3t-2n +\epsilon n\text,
  \end{align*}
  and hence
  \[ \abs{X\cap A(P)} = \abs{X} - \abs{X \cap (A(P)\cup B(P))^c}
  > 3n - 4t - 1/\epsilon - \epsilon n > \epsilon n\text, \]
  using $\abs{X} = \abs{R(P)}>n-t-1/\epsilon$.
\end{proof}

By Proposition \ref{prop:counting}, at most $1+\epsilon t$ vertices in $A(P)$ do not have a
$\ceil{1/\epsilon}$-successor in $A(P)$.
Thus at least $2\epsilon n -\epsilon t-1 > 0$ vertices $p_i\in A(P)$ satisfy the following properties:
\begin{enumerate}
\parskip 0pt
\item $c(p_i,p_{i+1}) \in c(p_1,R(P))$ and
\item $p_i$ has a $\ceil{1/\epsilon}$-successor $p_j\in A(P)$.
\end{enumerate}
Take any such vertices $p_i$ and $p_j$.

\begin{proposition}\label{prop:inter}
  We have \[\abs{\newn{p_i,P}\cap \newn{p_j,P} \cap V(P)} > \epsilon n\text.\]
\end{proposition}
\begin{proof}
  Because $p_i,p_j\in A(P)$,
  the paths $P_i = \rho_i \cdot P$ and $P_j = \rho_j \cdot P$ are maximum rainbow paths.

  Both paths $P_i$ and $P_j$ end in the same vertex $p_t$. It will be useful to
  consider $P_i$ and $P_j$ as directed paths, so let $E_i$ and $E_j$ be the
  edge sets of $P_i$ and $P_j$, but with the edges directed towards $p_t$.
  Because $p_j$ is a $\ceil{1/\epsilon}$-successor
  of $p_i$, we have
  \[ \abs{E_i\cup E_j} \leq t-1+\ceil{1/\epsilon} \quad \text{and}\quad \abs{E_i\cap E_j} \geq t-1-\ceil{1/\epsilon}\text. \]

  Since $P_i$ and $P_j$ are maximum rainbow paths, the colors in $c(p_t,V(P)^c)$
  must appear on edges of $P_i$ and $P_j$.
  Let $Z$ be the set of edges in $E_i\cap E_j$ that are colored with
  colors in $c(p_t,V(P)^c)$.
  We have \[\abs{Z} \geq \abs{c(p_t,V(P)^c)} - \ceil{1/\epsilon} = n-t-\ceil{1/\epsilon}\text.\]
  If we write $P_i = (x_1,\dotsc,x_t)$ and $P_j = (y_1,\dotsc,y_t)$, then we can define
  \[X = \{(x_i,x_{i+1}) : x_{i}\in A(P_i) \} \quad \text{and}\quad Y = \{(y_i,y_{i+1}) : y_{i}\in A(P_j) \}\text. \]
  Note that it is sufficient to show that $\abs{X\cap Y}> \epsilon n+2$.
  By maximality of $P_i$ and $P_j$,
  \[ Z \cap X = Z \cap Y = \emptyset\text.\]

  Therefore, we get
  \begin{align*} \abs{Z \cup X \cup Y}
    & = \abs{Z} + \abs{X \cup Y}\\
    & = \abs{Z} + \abs{X} + \abs{Y}
    - \abs{X \cap Y}\\
    & \geq 3n - 3t - \ceil{1/\epsilon}- \abs{X\cap Y}\text.
  \end{align*}

  But $Z \cup X\cup Y\subseteq E_i\cup E_j$, so we have
  \[ \abs{Z \cup X \cup Y} \leq \abs{E_i\cup E_j} \leq t+\ceil{1/\epsilon}\text,\]
  and therefore
  \[ \abs{X\cap Y} \geq 3n-4t-2\ceil{1/\epsilon}
  >\epsilon n+2 \text,\]
  which is what we needed to prove.
\end{proof}

In the following, let us write
\[ X = \newn{p_i,P}\cap \newn{p_j,P}\cap V(P)\text.\]
We have just shown that $\abs{X}\geq \epsilon n$.
Because there are less than $1/\epsilon$ vertices of $P$ between $p_i$ and $p_j$, we can say
that if $n$ is large enough, then either there are at least $\epsilon n /3$ vertices
of $X$ preceding $p_i$ on $P$ (Case 1), or there are $\epsilon n /3$
vertices of $X$ succeeding $p_j$ (Case 2).
Now we distinguish between the Cases 1 and 2.

\begin{enumerate}[{Case }1.]
\parskip 0pt
\item In this case, there are at least $\epsilon n/3$ vertices of $X$ preceding $p_i$.
  Using Proposition \ref{prop:counting}, there are at most $1+\epsilon t/10<\epsilon n/10$ vertices in $X$
  that do not have a $\ceil{10/\epsilon}$-successor in $X$.

  Therefore, at least
  \[ \frac{\epsilon n}{3}-\frac{\epsilon n}{10} = \frac{7\epsilon n}{30} \]
  of the vertices of $X$ preceding $p_i$ have a $\ceil{10/\epsilon}$-successor in $X$
  Note that there is a bijection between the vertices and their successors -- to every vertex in $X$
  corresponds his closest successor in  $X$.
  So by the same argument, of those successors, at least
  \[ \frac{7\epsilon n}{30} - \frac{\epsilon n}{10} = \frac{4\epsilon n}{30} \]
  have themselves a $\ceil{10/\epsilon}$-successor in $X$.

  Hence there is a vertex $u \in X$ preceding $p_i$
  on $P$ that has two $\ceil{20/\epsilon}$-successors in $X$ which also
  precede $p_i$. In particular, $u$ has a $\ceil{20/\epsilon}$-successor $v$ that precedes $p_i$
  and that satisfies $c(p_i,u)\neq c(p_j,v)$.

  Then the path shown in figure \ref{fig:case1}
  is a rainbow path starting in $p_1$ and ending in $p_t$, visiting at least $t-\ceil{1/\epsilon}-\ceil{20/\epsilon}$ vertices,
  that uses only two colors not in $c(P)$.

\item As in the previous case, there is a vertex
  $v \in X$ succeeding $p_j$ on $P$
  that has a $\ceil{20/\epsilon}$-successor $v$ which satisfies $c(p_i,u)\neq c(p_j,v)$.
  The path shown in figure \ref{fig:case2}
  is a rainbow path starting in $p_1$ and ending in $p_t$, visiting at least $t-\ceil{1/\epsilon}-\ceil{20/\epsilon}$ vertices,
  that uses only two colors not in $c(P)$.
\end{enumerate}

\begin{figure}
  \centering
  \begin{pspicture}(-5.2,-2)(5.2,2)
    \psset{fillstyle=solid,fillcolor=black}
    \cnode(-5,0){.08}{p1}  \nput{-90}{p1}{$p_1$}
    \cnode(-2.5,0){.08}{u}  \nput{-90}{u}{$u$}
    \cnode(-1.5,0){.08}{v} \nput{-90}{v}{$v$}
    \cnode(1.5,0){.08}{pi}   \nput{-90}{pi}{$p_i$}
    \cnode(2.5,0){.08}{pj}   \nput{-90}{pj}{$p_{j}$}
    \cnode(5,0){.08}{pt}   \nput{-90}{pt}{$p_t$}

    \pnode(-4,0){a}
    \ncline{p1}{a}
    \pnode(-3.5,0){b}
    \ncline[nodesep=3pt,linestyle=dotted]{a}{b}
    \ncline{b}{u}

    \pnode(3.5,0){a1}
    \ncline{pj}{a1}
    \pnode(4,0){b1}
    \ncline[nodesep=3pt,linestyle=dotted]{a1}{b1}
    \ncline{b1}{pt}

    \psset{fillstyle=none,arcangle=45}
    \ncline{v}{pi}
    \ncarc{u}{pi}
    \ncarc{pj}{v}
  \end{pspicture}
  \caption{Figure for Case 1}
  \label{fig:case1}
\end{figure}

\begin{figure}
  \centering
  \begin{pspicture}(-5.2,-2)(5.2,2)
    \psset{fillstyle=solid,fillcolor=black}
    \cnode(-5,0){.08}{p1}  \nput{-90}{p1}{$p_1$}
    \cnode(-2.5,0){.08}{u}  \nput{-90}{u}{$p_i$}
    \cnode(-1.5,0){.08}{v} \nput{-90}{v}{$p_j$}
    \cnode(1.5,0){.08}{pi}   \nput{-90}{pi}{$u$}
    \cnode(2.5,0){.08}{pj}   \nput{-90}{pj}{$v$}
    \cnode(5,0){.08}{pt}   \nput{-90}{pt}{$p_t$}

    \pnode(-4,0){a}
    \ncline{p1}{a}
    \pnode(-3.5,0){b}
    \ncline[nodesep=3pt,linestyle=dotted]{a}{b}
    \ncline{b}{u}

    \pnode(3.5,0){a1}
    \ncline{pj}{a1}
    \pnode(4,0){b1}
    \ncline[nodesep=3pt,linestyle=dotted]{a1}{b1}
    \ncline{b1}{pt}

    \psset{fillstyle=none,arcangle=45}
    \ncline{v}{pi}
    \ncarc{u}{pi}
    \ncarc{pj}{v}
  \end{pspicture}
  \caption{Figure for Case 2}
  \label{fig:case2}
\end{figure}

Notice that by this construction, the path that we get is always such that
some color in $c(p_1,R(P))$ is not used, namely, $c(p_i,p_{i+1})$.
But since $t-\ceil{1/\epsilon}-\ceil{20/\epsilon} > t-a/3$, in both cases the resulting path violates
Lemma \ref{lemma:no-paths}. This is the desired contradiction.

\section{Conclusion}

We have proved that for every $\epsilon > 0$, every proper edge coloring of the graph $K_n$
contains a rainbow path on
\[t > \left(\frac{3}{4}-\epsilon\right)n\]
vertices, assuming that $n$ is larger than some value $n_0$ depending on $\epsilon$.
This is a significant improvement over Theorem \ref{thm:gm2010}.

In~\cite{Akbari2007}, the authors proved that every proper edge coloring of $K_n$ contains Hamiltonian
cycles on at least $(2/3-o(1))n$ different colors.
Clearly, the extension of a rainbow path of length $t$ to a Hamiltonian cycle in $K_n$ gives
a Hamiltonian cycle using at least $t$ colors.
So as a bonus we get the following corollary to Theorem \ref{thm:awesome}.
\begin{corollary}
  In every proper edge coloring of $K_n$, there are Hamiltonian cycles using at least
  $(3/4-o(1))n$
  different colors.
\end{corollary}

\chapter{Paths with Repeated Colors}

\section{Introduction}

In this chapter, we take a look at a natural generalization of rainbow paths.
Given a proper edge coloring of $K_n$, a rainbow path uses no color more than once;
now we allow for \emph{$k$-rainbow paths}, using every color at most a fixed number $k$ of times.
Note that $1$-rainbow paths are just rainbow paths, so we have the following theorem.

\begin{thm}[{\cite{GyarfasMhalla2010}}]
  \label{thm:1-rainbow}
  In every proper edge coloring of $K_n$, there is a $1$-rainbow path on at least
  $(2n+1)/3$ vertices.
\end{thm}

Mostly, we are interested in an asymptotic statement:  how does the length of a maximum
$k$-rainbow path increase with $k$? The following is a simple bound, although for simplicity
we only prove it for powers of two.

\begin{proposition}
  \label{prop:naive}
  Let $n = 2^m$, for some $m \in \naturals$.
  In every proper edge coloring of $K_n$ and for any $k\geq 1$,
  there is a $k$-rainbow path of length $(1-1/2^k)n - O(k)$.
\end{proposition}
\begin{proof}
  We actually prove a stronger statement: that for every vertex $v$ of $K_n$, there is a $k$-rainbow path starting
  in $v$ of length $(1-1/2^k)n - O(k)$.
  The proof goes by induction on the number of vertices. In the base case of $K_1$, there are no paths of nonzero length, so the
  claim is trivially satisfied.

  Now assume that the claim is true for some value $n$ and consider the graph $K_{2n}$.
  Let $v$ and $k$ be given. Starting in $v$, there is a rainbow path $P=(v,p_2,\dotsc,p_n)$ on
  $n$ vertices in $K_{2n}$. This was noted after the proof of Proposition \ref{prop:n/2-bound} in
  the previous chapter.
  Then we can invoke the induction hypothesis
  to get a $(k-1)$-rainbow path starting in $p_n$ and avoiding the vertices of $P$, of length
  $(1-1/2^{k-1})n-O(k-1)$.
  Appending the two paths together, we get a $k$-rainbow path of length
  \[ n-1 + \left(1-\frac{1}{2^{k-1}}\right)n - O(k-1) = \left(1-\frac{1}{2^k}\right)
  2n - O(k)\text, \]
  concluding the proof.
\end{proof}

This proof, while simple, already contains an important idea: that we can use $(k-1)$-rainbow paths
to build $k$-rainbow paths.
The rest of this chapter is about improving on this result, but first we define some notation.

Consider the complete graph $K_n = (V,E)$ and a proper edge coloring $c$ of $K_n$.
If $P$ is a $k$-rainbow path with respect to $c$, then let $C_0(P) = c(P)^c$ and for
$i\in [k]$, let $C_i(P)$ be the set of colors used exactly $i$ times
on edges of $P$. Clearly, $c(E) = \bigcup_{i=0}^k C_i(P)$.

Furthermore, for a $k$-rainbow path $P=(p_1,\dotsc,p_t)$, we define
\[C_A(P) = \{c(p_{i},p_{i+1}) : c(p_1,p_{i+1}) \nin C_k(P)\}\text.\]

Now we define what we mean by a maximal $k$-rainbow path.
A $k$-rainbow path $P=(p_1,\dotsc,p_t)$ is \emph{maximal} if it satisfies both
\begin{gather}
  c(p_1, V(P)^c) \subseteq C_k(P)\text.\label{eq:maximality}\\
  \shortintertext{and}
  c(p_t,V(P)^c) \subseteq C_k(P) \setminus C_A(P)
  \label{eq:maximality-rot}
\end{gather}
We shall see that every maximum length $k$-rainbow path is also
maximal in this sense.

\section{Two lemmas on maximal $k$-rainbow paths}

\begin{lemma}
  If $P$ is a $(k-1)$-rainbow path, then
  there is a maximal $k$-rainbow path $P'$ with
  $\abs{C_{k}(P')} \leq \abs{V(P')}-\abs{V(P)}$.
\end{lemma}
\begin{proof}
  Let $P$ be any $(k-1)$-rainbow path. Then
  $P$ is also a (non-maximal) $k$-rainbow path with $C_{k}(P) = 0$.
  We will show that for any non-maximal $k$-rainbow path $P$, there is a $k$-rainbow
  path $P'$ with $\abs{V(P')}=\abs{V(P)}+1$ and $\abs{C_k(P')}\leq \abs{C_k(P)} +1$.
  Since we cannot add vertices indefinitely, we will
  eventually get a maximal $k$-rainbow path with the required properties.

  So if $P= (p_1,\dotsc,p_t)$ is a non-maximal $k$-rainbow path, then
  one of the following is the case.
  \begin{enumerate}[{Case }1.]
    \parskip0pt
    \item If $c(p_1, V(P)^c) \not\subseteq C_{k}(P)$, then
      there is an edge $\{p_1,r\}\in E(p_1,V(P)^c)$ colored with a color $c\nin C_{k}(P)$.
      Hence the path $P \cup \{p_1,r\}$ is a $k$-rainbow path and
      $C_{k}(P\cup\{p_1,r\}) = C_{k}(P)+1$.

    \item If $c(p_t, V(P)^c) \not\subseteq C_{k}(P)$, then
      we can proceed just as in the first case,
      after reversing the order of the vertices on $P$.

    \item If $c(p_t, V(P)^c) \subseteq C_{k}(P)$ but $c(p_t, V(P)^c) \not\subseteq C_{k}(P)\setminus C_A(P)$,
      then there are vertices $p_i\in V(P)$ and $r\in V(P)^c$ such that $c(p_i,p_{i+1})=c(p_t,r)\in C_k(P)$
      and, furthermore, $c(p_1,p_{i+1})\nin C_{k}(P)$.

      Recall that $\rho_i \cdot P$ is the path $(p_i,p_{i-1},\dotsc,p_1,p_{i+1},\dotsc,p_t)$.
      Then $\rho_i \cdot P$ is a $k$-rainbow path with $\abs{C_k(\rho_i\cdot P)}\leq \abs{C_k(P)}$
      and $\abs{V(\rho_i\cdot P)} = \abs{V(P)}$.
      Moreover, $\rho_i\cdot P$ ends in $p_t$ and we have $c(p_t,r)\in C_{k-1}(\rho_i\cdot P)$.
      This means that
      $c(p_t, V(\rho_i\cdot P)^c) \not\subseteq C_{k}(\rho_i\cdot P)$, and
      so can we proceed as in the second case.
  \end{enumerate}
  This completes the proof of the lemma. Note that in addition, we have proved that every maximum
  length $k$-rainbow path is maximal.
\end{proof}

\newcommand{\edges}[1]{\mathcal{E}[#1]}
\begin{lemma}
  \label{lemma:ck}
  If $P$ is a maximal $k$-rainbow path, then
  \[ \abs{C_k(P)} \geq (k+1)n - (k+1)\abs{V(P)}\text.\]
\end{lemma}
\begin{proof}
  Let $P=(p_1,\dotsc,p_t)$ be a maximal $k$-rainbow path.
  If $C\subseteq C(E)$ is a set of colors, then we write
  \[ \edges{C} = \{ e\in E(P) : c(e)\in C \} = c^{-1}(C)\cap E(P) \]
  for the set of edges of $P$ colored with a color in $C$.

  First, we would like to find a lower bound for $\abs{C_A(P) \cap C_k(P)}$.
  Since every color appears at most $k$ times on $P$,
  we have
  \[ k\abs{C_A(P)\cap C_k(P)} \geq \abs{\edges{C_A(P)\cap C_k(P)}}\text. \]
  By maximality condition \eqref{eq:maximality},
  every vertex $v$ with $c(p_1,v) \nin C_k(P)$ is in $V(P)$,
  so we have
  \[ \abs{\edges{C_A(P)}}\geq \abs{C_A(P)} \geq n-1-\abs{C_k(P)}\text. \]
  Moreover,
  \[ \abs{\edges{C_k(P)^c}} = \abs{E(P)}-\abs{\edges{C_k(P)}} = t-1-k\abs{C_k(P)}\text.  \]
  Then we get
  \begin{align*} k\abs{C_A(P)\cap C_k(P)} &\geq \abs{\edges{C_A(P)\cap C_k(P)}}\\
    & = \abs{\edges{C_A(P)}\setminus \edges{C_k(P)^c}}\\
    &\geq \abs{\edges{C_A(P)}} - \abs{\edges{C_k(P)^c}}\\
    &\geq n-1-\abs{C_k(P)} - (t-1-k\abs{C_k(P)})\\
    &= n-t +(k-1)\abs{C_k(P)}\text.
  \end{align*}

  By maximality condition \eqref{eq:maximality-rot},
  \begin{align*} k\abs{c(p_t,V(P)^c)} &\leq k\abs{C_k \setminus C_A(P)} \\
    &= k\abs{C_k(P)} - k\abs{C_A(P) \cap C_k(P)} \\
    &\leq k\abs{C_k(P)} - n + t - (k-1)\abs{C_k(P)}\\
    &= \abs{C_k(P)} - n + t
  \end{align*}

  With $\abs{c(p_t,V(P)^c)} = n-t$, we get
  \[ \abs{C_k(P)} \geq (k+1)n - (k+1)t\text, \]
  as claimed.
\end{proof}

\section{A theorem on the length of $k$-rainbow paths}

\begin{thm}
  \label{thm:k-bound}
  In every proper coloring of $K_n$ and for any $k\geq 1$,
  there is a $k$-rainbow path on at least
  \[ \left(1-\frac{2}{(k+2)!}\right)n \]
  vertices.
\end{thm}
\begin{proof}
  The proof goes by induction on $k$.

  The induction basis is provided for by Theorem \ref{thm:1-rainbow},
  since \[ \frac{2n+1}{3} \geq \left(1-\frac{2}{3!}\right)n\text.\]

  In the induction step, assume that there is a
  $(k-1)$-rainbow path using
  \[t_{k-1} > \left(1-\frac{2}{(k+1)!}\right)n\] vertices.
  Then, by Lemma 1, there is a maximal $k$-rainbow path $P$ with $\abs{C_k(P)}\leq \abs{V(P)}-t_{k-1}$.
  Using Lemma 2, we get
  \[ \abs{C_k(P)} \geq (k+1)n - (k+1)\abs{V(P)}\text,\]
  so if we write $t_k$ for $\abs{V(P)}$, then
  \[ t_k-t_{k-1} \geq (k+1)n - (k+1)t_k\text,\]
  or
  \[ (k+2)t_k \geq t_{k-1} + (k+1)n\text.\]

  Using the induction hypothesis, we get
  \begin{align*}
    t_k& \geq \frac{t_{k-1} + (k+1)n}{k+2}\\
    &> \frac{\left(1-\frac{2}{(k+1)!}\right)n}{k+2} + \frac{(k+1)n}{k+2}\\
    &= \frac{n}{k+2} + \frac{(k+1)n}{k+2} - \frac{2n}{(k+2)!}\\
    &= \frac{(k+2)n}{k+2} - \frac{2n}{(k+2)!}\\
    &= \left(1 - \frac{2}{(k+2)!}\right)n\text,
  \end{align*}
  so $P$ is a $k$-rainbow path of sufficient length.
\end{proof}

\section{Conclusion}

From the statement of Theorem \ref{thm:k-bound}, it is easily seen that for fixed $n$,
the number of vertices not included in a maximum $k$-rainbow path is in the order of $1/k!$.
This is clearly an improvement over Proposition \ref{prop:naive}, which only
shows that the number
of vertices not included in a maximum $k$-rainbow path is in the order of $1/2^k$.
The growth provided by Theorem \ref{thm:k-bound} is asymptotically faster.

In the proof, we essentially used a generalized version of the argument
made by Gyárfás and Mhalla in~\cite{GyarfasMhalla2010}, a discussion of which can be found in
the previous chapter.
In this light, it might be interesting to generalize the techniques used in the proof of
the bound of $(3/4-o(1))n$ for the length of maximum rainbow paths in $K_n$, and to
apply them to $k$-rainbow paths.

\chapter{Conclusion}

In the preceding chapters, we have derived two novel results on the existence of certain paths
in proper edge colorings of the complete graph $K_n$.

Most importantly, we have shown that in every proper edge coloring of $K_n$ there is a
rainbow path of length at least \[\left(\frac{3}{4}-o(1)\right)n\text.\]
This result improves on the previously best known bound of $2n/3$ proved by
Gyárfás and Mhalla in~\cite{GyarfasMhalla2010}.
As a corollary, there are Hamiltonian cycles using at least $(3/4-o(1))n$ colors;
here we improve on a result by Akbari, Etesami, Mahini and Mahmoody~\cite{Akbari2007}.

Moreover, we have proved that in every proper edge coloring of $K_n$,
there is a $k$-rainbow path on at least \[ \left(1-\frac{2}{(k+2)!}\right)n \]
vertices, for any $k>0$.
Thus, for fixed $n$, the number of vertices not included in a maximum $k$-rainbow
path decreases with $k$
faster than any exponential function, and hence asymptotically faster than what
we got using a naive approach.

We believe that the techniques used in the proof
of the first result, which relied heavily on pigeonhole-style arguments,
do not immediately lend themselves to proving stronger bounds.
Moreover, the proof itself does not seem to reveal deep insights into
the structure of rainbow paths, leading us to believe that different
methods will have to be used to prove the existence of rainbow paths of length $n-o(n)$.

However, it seems likely that applying the same methods, suitably generalized,
to the problem of $k$-rainbow paths in $K_n$ might prove to be fruitful. Indeed, we would
expect the resulting bound to be asymptotically stronger (in $k$, for fixed $n$)
than our bound, whose proof relied on comparatively simple techniques.

\backmatter

\bibliographystyle{alpha}
\bibliography{refs}

\end{document}